\documentclass[journal]{IEEEtran}
\ifCLASSINFOpdf
\else
\fi
\usepackage{subfigure} 
\usepackage{graphicx}
\usepackage{amssymb}
\usepackage{amsmath}
\usepackage{color} 
\usepackage{cite}
\usepackage{url} 
\usepackage{array,tabularx}
\usepackage{multicol} 
\usepackage{multirow} 
\usepackage{makecell}
\usepackage{diagbox}
\usepackage{booktabs}
\newtheorem{theorem}{Theorem}
\newtheorem{proposition}{Proposition}

\newtheorem{lemma}{Lemma}
\newtheorem{definition}{Definition}
\newtheorem{assumption}{Assumption}
\newtheorem{remark}{Remark}
\newtheorem{proof}{Proof}

\newtheorem{property}{Property}
\usepackage{tikz,xcolor,hyperref}
\definecolor{lime}{HTML}{A6CE39}
\DeclareRobustCommand{\orcidicon}{
	\begin{tikzpicture}
		\draw[lime, fill=lime] (0,0) 
		circle [radius=0.16] 
		node[white] {{\fontfamily{qag}\selectfont \tiny ID}}; 
		\draw[white, fill=white] (-0.0625,0.095) 
		circle [radius=0.007];	  
	\end{tikzpicture}
	\hspace{-2mm}}
\foreach \x in {A, ..., Z}{%
	\expandafter\xdef\csname orcid\x\endcsname{\noexpand\href{https://orcid.org/\csname orcidauthor\x\endcsname}{\noexpand\orcidicon}}
}
\makeatother

\begin{document}
\title{A Distributed PI+Reset Scheme for Discrete-Time Economic Dispatch of A Grid-connected BESS Network}
\author{
	Yalin Zhang\orcidA{},
        Zhongxin Liu\orcidB{},
        Fuyong Wang\orcidC{},
        and Zengqiang Chen\orcidE{}
\thanks{
	Manuscript received XX, XX; revised XX, XX;
	accepted XX, XX. This work was supported in part by the National Natural Science Foundation of China (Grant No. 62103203 and 92367105) and the General Terminal IC Interdisciplinary Science Center of Nankai University (Corresponding author: Fuyong Wang.). 
	\par The authors are with the College of Artificial Intelligence, Nankai University, Tianjin 300350, and also with the Tianjin Key Laboratory of Interventional Brain-Computer Interface and Intelligent Rehabilitation, Nankai University, Tianjin 300350, China (e-mail: zhangyl@mail.nank.edu.cn; lzhx@nankai.edu.cn; wangfy@nankai.edu.cn; zhangym@mail.nankai.edu.cn; chenzq@nankai.edu.cn).
}
}
\maketitle
\begin{abstract}
	This article investigates the discrete-time economic dispatch (ED) problem of a battery energy storage system (BESS) network with an energy router (ER). The continuous increase in operational cost of a BESS network is caused by the internal power consumption and capacity degradation of each battery. In addition, the transaction amount of purchasing or selling electricity from the utility grid (UG) also becomes one of the sources that constitute this cost. Therefore, in order to address this ED problem and reduce costs, we design a distributed solution based on discrete-time multi-agent systems (MAS) with a novel proportional integral (PI) controller. In this scheme, a marginal cost (MC) consensus controller is designed to drive the inverter. In addition, a consensus controller is designed to estimate the average power mismatch, resulting in a routing algorithm based on this. Compared with existing distributed schemes with proportional (P) controllers, using a PI controller with a reset mechanism ensures that the integral term accumulates from 0 when the proportional term changes sign. Driven by this method, the convergence speed of the scheme is accelerated, while the control accuracy is also improved without causing significant overshoot. Provided the enabling conditions for the reset mechanism and analyzed the algorithm performance under SoC level constraints. The related simulation cases verify the effectiveness and progressiveness of the designed algorithm.
\end{abstract}
\section{Introduction}
\subsection{Background}
\IEEEPARstart{R}{enewable} energy generation exhibits intermittency and randomness, resulting in low power quality and even power outages \textcolor{blue}{\cite{wangApplicationEnergyStorage2022, ANSARI2024123996}}. Fortunately, the battery energy storage system (BESS) can suppress the power supply and demand gap and power fluctuations, thus helping to alleviate power shortages and improve power supply quality \textcolor{blue}{\cite{caleroReviewModelingApplications2022, 10916978}}. In view of this, BESSs are used to be integrated into microgrids, as shown in Fig. \ref{figei}, and assist in providing reliable electricity \cite{wangApplicationEnergyStorage2022, caleroReviewModelingApplications2022, 10916978, hossainlipuReviewControllersOptimizations2022}. However, for a grid-connected BESS network, the energy storage loss of each battery unit and the power transmission loss increase the operating cost of a BESS network \cite{caleroReviewModelingApplications2022, forero-quinteroProfitabilityAnalysisDemandside2022, rouholaminiReviewModelingManagement2022}. In addition, the transaction volume between this network and the utility grid (UG) has also become a part of daily operating costs. 
It will inevitably bring huge operating costs to the operators of this power network if there is no efficient and reliable ED solution.
\begin{figure}
	\centering
	\includegraphics[width=8.5cm]{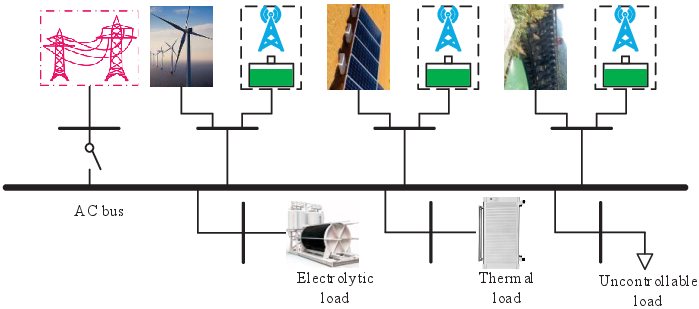} \caption{BESSs in microgrids\label{figei}}
\end{figure}
\par For the ED problem of microgrids, some  centralized solutions are initially centralized, each one of which is simple and feasible, but susceptible to single point failure and requires high communication bandwidth and high cost control centers \cite{caleroReviewModelingApplications2022, rouholaminiReviewModelingManagement2022}. These defects indicate that a centralized scheme is not an efficient and reliable solution.  Recently, we note that distributed control technology, due to its robustness and scalability, can avoid the use of control centers and high bandwidth communication links, thus greatly reducing design costs and improving reliability. Therefore, we plan to invest effort in designing an efficient distributed ED scheme for a grid-connected BESS network.
\subsection{Literature Review and Motivation}
\par With the rise of distributed technology represented by multi-agent systems (MASs), there have been significantly improved in solutions to the ED problem of microgrids. For example, authors in \cite{Yan2022a} design a distributed ED algorithm for microgrids with an event triggered communication mechanism and a privacy protection mechanism. Similarly, the ED problem of microgrids is solved by a distributed ED scheme with an event triggering communication mechanism in \cite{Wan2021} under uncertain communication. In response to the phenomenon of packet loss in communication, authors in \cite{Li2023a} design a distributed algorithm for power mismatch estimation and exchange power by introducing establishing virtual buffer nodes. Furthermore, authors in \cite{Liu2022b} design a distributed finite-time scheme for the ED problem of microgrids with an event triggered communication scheme. To further improve the convergence rate, authors in \cite{Dai2021, Li2022m} design a distributed fixed-time ED scheme, which are further developed as a solution that can suppress interference and support intermittent communication in \cite{Huang2023b}. Besides, authors in \cite{Mao2021} transforms the ED problem into its dual problem and design a finite-time marginal cost (MC) consensus protocol with local gradients, which improves in \cite{Ji2023} by applying the prescribed-time control theory. All of these results can be attributed to the design of two distributed schemes, namely MC consensus controller \cite{tanExtensionsLocationalMarginal2022, spanglerPowerGenerationOperation2014} and a power mismatch estimator. 
\par  A MC consensus controller has other significant uses. For example, based on a MC controller, authors in \cite{zaeryNovelFullyDistributed2021} design economic droop control for DC microgrids, and the fixed time control theory is combined to achieve secondary control objectives. Authors in \cite{Li2022l} design an economic droop based on MC consensus and a communication strategy with a dynamic event triggering mechanism for an AC and DC hybrid microgrid to achieve secondary control. Besides, authors in \cite{Cheng2023} design a distributed algorithm for voltage restoration and ED problem of droop-controlled DC microgrids, taking into account power limitations. Furthermore, a distributed MC consensus controller capable of withstanding false data injection attacks is developed in \cite{Zhang2021d} to ensure the achievement of secondary control objectives. Some researchers in \cite{Chen2023a} design an MC controller that can solve non periodic sampling and time delay. 
\par These results show excellent performance in saving communication resources \cite{Yan2022a, Wan2021, Liu2022b, Huang2023b, Li2022l}, enhancing the reliability of communication networks \cite{Wan2021, Li2023a, Huang2023b, Zhang2021d}, and improving convergence rate \cite{Liu2022b, Dai2021, Li2022m, Huang2023b, zaeryNovelFullyDistributed2021}, respectively. Some of the above schemes \cite{Yan2022a, Wan2021, Li2023a, Mao2021, Li2022l, Zhang2021d} are unexceptionally with proportional (P) control controllers, which are effective, but  encounter low convergence speed and control
accuracy. However, solutions to accelerate convergence speed, such as finite/fixed time controllers \cite{Liu2022b, Dai2021, Li2022m, Huang2023b, Ji2023, zaeryNovelFullyDistributed2021}, often suffer from oscillations due to their fractional powers.
\par Drawing inspiration from ED schemes of microgrids, some researchers start to develop distributed ED schemes for a BESS network in \cite{yuFrequencySynchronizationPower2021b, zhaoDifferentialPrivacyEnergy2022, jinManageDistributedEnergy2022}. The authors in \cite{yuFrequencySynchronizationPower2021b} design a distributed scheme based on MC consensus controller and power mismatch estimation for BESS in isolated microgrids, taking into account battery capacity limitations. Unfortunately, the power loss on the transmission line is not considered, and the cost function used directly follows other distributed power sources, which means that the operating cost of BESS is not modeled. Subsequently, a distributed ED scheme considering privacy protection is designed for an isolated BESS network in \cite{zhaoDifferentialPrivacyEnergy2022}. However, the modeling problem of cost function of a BESS remains unresolved. Until \cite{jinManageDistributedEnergy2022}, the authors consider the cost of charging and discharging in the cost function.
\par The above results preliminarily explore some distributed ED schemes of a BESS network based on MC consensus controllers and power mismatch estimators. However, these schemes are only applicable to the ED problem in isolated mode. In addition, these schemes are all with P controllers, which often result in slow convergence speed and low control accuracy. Besides, the operating cost model of a BESS network on a small time scale is too simple and does not consider factors such as capacity degradation. Therefore, there is an urgent need for a cost function model that considers multiple factors and a distributed ED scheme with high control accuracy and convergence speed suitable for a grid connected BESS network.
\subsection{State of Contribution}
\par We have learned that the proportional integral (PI) controller exhibits high control accuracy \cite{Yi2016}, which is widely adopted in industrial production. But with it, the system tends to exhibit overshoot \cite{banosResetControlSystems2011}. In a groundbreaking work \cite{Yi2016}, it is required that the integral term is transmitted in the communication network, which is improved by \cite{10149512}. However, this solution aims to solve the power allocation problem in an isolated BESS network and does not consider transmission loss, requiring prior knowledge of the total load power in the network. Recently, some PI controllers improved with some reset mechanisms, namely PI+Reset controllers, have been developed to improve the dynamic and steady-state performance of distributed controllers for MASs with continuous time dynamics \cite{chengResetControlLeaderfollowing2022, huResetControlConsensus2022, mengResetControlSynchronization2019}. Based on these results, we have also preliminarily developed a distributed MC consensus controller with a continuous time dynamic for a grid-connected BESS network \cite{10302354}. In this paper, considering transmitted loss, we plan to build a discrete-time distributed ED scheme with PI+Reset controllers oriented to a grid-connected BESS network. 
In detail, our contributions are as follows:
\begin{enumerate}
	\item In the designed ED scheme, a PI+Reset controller is designed to promote MC consensus, where each BESS inverter is driven by its MC. The introduction of an integrator with a reset mechanism accelerates convergence and avoids overshoot, which is a significant advantage compared to previous distributed ED schemes with P controllers \cite{chenDistributedEconomicDispatch2021}.
	\item A routing scheme is designed based on a distributed average power mismatch estimator. The estimated average power mismatch by each agent converges to 0 faster and smoother under a PI+Reset controller with a further improved reset mechanism than under a P controller \cite{chenDistributedEconomicDispatch2021}. The average power mismatch estimated values collected by the neighbors of the agent responsible for managing ER are used to construct the routing scheme to quickly maintain the supply-demand balance.
	\item The enabling conditions of the reset mechanism and the stability of the PI+R controllers are well analyzed. By analyzing the eigenvalues of the closed-loop error system, the gain conditions for the effectiveness of the reset mechanism and stability are provided.
\end{enumerate}
\par Next, the relevant preliminaries are introduced in section II. The design scheme and discussions on their properties are organized in section III. Relevant simulation cases are investigated in section IV. Finally, a conclusion on this paper is drawn in section V.
\section{Preliminaries}
In this section, the cost and optimal ED of a BESS network are analyzed to construct the control objectives. The relevant knowledge of graph theory is introduced for future use. In the cost analysis, the internal power consumption and capacity degradation costs of a battery, power loss on the transmission line, and transaction volume between the network and UG in the grid-connected mode are considered. Meanwhile, the Lagrange multiplier method is applied to derive the optimal ED condition.
\subsection{The Optimal ED for A BESS Network}
\par In this section, we analyze the optimal ED for a BESS network with $n$ BESSs. Firstly, an operational cost function \cite{10302354} for BESS $i$ is introduced as
$$f_i=\beta_iP_i^2+\alpha_iP_i,\;i\in\{1,2,\cdots,n\},$$
where $\beta_i$ and $\alpha_i$ are constants.
\par For a BESS network connected to UG, the shortage in power will be purchased from the company managing UG. On the contrary, excessive electricity will also be sold to the company. This transaction can be modeled as
$$f_{UG}=\rho P_{UG},$$
where $\rho$ is the transaction electricity price agreed upon by both parties, $P_{UG}$ is the exchange power. Here, a unified electricity price is adopted for the electricity trading. $P_{UG}>0$ implies power flows from UG to the BESS network, and $P_{UG}<0$ implies that the power flows in the opposite direction.
\par Considering the power loss on the transmission line, a general simplified model \cite{chenDistributedEconomicDispatch2021, Soliman2012} derived from the B matrix representation \cite{spanglerPowerGenerationOperation2014}, which encompasses the topology information of a power network, is as follows
$$P_{L,i}=Z_iP_i^2,$$
where $P_{L,i}$ is the transmission loss caused by the output power of BESS $i$, $Z_i$ is the output impedance of BESS $i$. 
\par Thus, the operation cost function and its constraints of a grid-connected BESS network is formulated as
$$\begin{array}{l}
	\mathrm{min}\;F=\sum_{i=1}^{n}f_i+f_{UG},\\
	s.t.\;\sum_{i=1}^{n}P_i+P_{UG}=\sum_{i=1}^{n}(P_{D,i}+P_{L,i}+P_{DG,i}),\\
	\;\;\;\;\;\;P_{m,i}\le P_i\le P_{M,i},
\end{array}$$
where $P_{m,i}$ and $P_{M,i}$ are the minimum and maximum output power limits of BESS $i$, respectively. In the following text, regardless of the presence or absence of BESS on bus $i$, $P_i$ is referred to as the output power of BESS $i$. It should be emphasized that for buses without any BESS, such as bus $j$, $P_{M,j}=0$ and $P_{m,j}=0$.
\par Construct a Lagrangian function for a grid-connected BESS network as follows
$$\begin{array}{l}
	L_a(P_1,\cdots,P_n,P_{UG})=F+\omega_1(\sum_{i=1}^{n}(P_{D,i}+P_{L,i})\\
	-\sum_{i=1}^{n}P_i-P_{UG})+\sum_{i=1}^{n}\omega_{2,i}(P_i-P_{M,i})\\
	+\sum_{i=1}^{n}\omega_{3,i}(P_{m,i}-P_i),
\end{array}$$
where $\omega_1$, $\omega_{2,i}$, and $\omega_{3,i}$ are Lagrangian multipliers for BESS $i$, respectively. Without considering capacity limitations, calculate the gradient of the above as follows
$$\frac{\partial L_a(P_1,\cdots,P_n,P_{UG})}{\partial P_i}=2\beta_iP_i+\alpha_i+\omega_1(2Z_i P_i-1),$$
$$\frac{\partial L_a(P_1,\cdots,P_n,P_{UG})}{\partial P_{UG}}=\rho-\omega_1,$$
which leads to the optimal condition as follow
\begin{equation}
	\label{gsmmc}
	\left\{\begin{array}{l}
		\omega_1^*=\frac{2\beta_iP_i+\alpha_i}{1-2Z_i P_i},\\
		\omega_1^*=\rho.
	\end{array}\right.
\end{equation}
Combined with the capacity limitations, the optimal output power of BESS $i$ can be calculated from the following, i.e.,
\begin{equation}
	\label{power}
	\textcolor{blue}{P_i^*=\left\{\begin{matrix}
			P_{m,i},&if\;\frac{\omega_1^*-\alpha_i}{2(\beta_i+Z_i\omega_1^*)}<P_{m,i},\\
			\frac{\omega_1^*-\alpha_i}{2(\beta_i+Z_i\omega_1^*)},&if\;P_{m,i}\le\frac{\omega_1^*-\alpha_i}{2(\beta_i+Z_i\omega_1^*)}\le P_{M,i},\\
			P_{M,i},&if\;\frac{\omega_1^*-\alpha_i}{2(\beta_i+Z_i\omega_1^*)}>P_{M,i},
		\end{matrix}
		\right.}
\end{equation}
and the optimal exchange power $P_{UG}$ is
\begin{equation}
	\label{powermg}
	P_{UG}^*=\sum_{i=1}^{n}(P_{D,i}+P_{L,i}^*-P_i^*),
\end{equation}
where $P_{L,i}^*=Z_i(P_i^*)^2$, and $\sum_{i=1}^{n} P_{L,i}^*$ is so-called the optimal line loss power.
\subsection{Control Objectives}
\par The goals of the ED for a BESS network are often divided into two parts, where one is to minimize the total power generation cost, and the other is to ensure a balance between supply and demand. Specifically, as shown in Fig. \ref{figei}, it is necessary to design a routing protocol for the ER and a control scheme for each BESS inverter to meet the optimal ED conditions. Based on the previous analysis, we have summarized these goals in detail as follows
\begin{itemize}
	\item Define MC for BESS $i$ as $\lambda_i=\frac{2\beta_iP_i+\alpha_i}{1-2Z_i P_i}$.
	For the grid-connected mode, MCs converge asymptotically to the real-time electric price $\rho$ of UG, i.e.,
	$$\lim\limits_{k\to+\infty}\vert\lambda_i-\rho\vert,\;\forall i\in\{1,2,\cdots,n\},$$
	and calculates the optimal power flow by \eqref{power} to drive each inverter.
	\item The total demand is defined here as the sum of the total load and the bus loss power. For the grid-connected mode, a routing protocol should be developed to drive the ER. Calculate the power exchange between this BESS network and UG by \eqref{powermg} to ensure a balance between power supply and demand, i.e.,
	$$\sum_{i=1}^{n}(D_i+P_{L,i}-P_i-P_{UG})=0,$$
	where $D_i$ is the load at bus $i$.
\end{itemize}
\subsection{Graph Theory}
\par For MASs containing $n$ agents that manages BESSs and a ER, the communication topology between agents can be described by a graph ${\cal G}({\cal V},{\cal E})$ containing a vertex set ${\cal V}=\{v_1,v_2,\cdots,v_n\}$ and an edge set ${\cal E}$. If agent $i$ can access agent $j$, it is assumed that $(v_i,v_j)\in{\cal E}$, and agent $j$ is a neighbor
of agent $i$. Define an adjacency matrix and a degree matrix as ${\cal A}=[a_{ij}]_{n\times n}$ and ${\cal D}=diag(d_1,d_2,\cdots,d_n)$, respectively. $a_{ij}=1$ if $(v_i,v_j)\in{\cal E}$ and $a_{ij}=0$ otherwise. $d_i$ is the in-degree of agent $i$, i.e., $d_i=\sum_{j=1}^{n}a_{ij}$. Thus, a Laplacian matrix is defined as $L={\cal D}-{\cal A}$.
\par The directed graph and its corresponding Laplace matrix are introduced above. It is easy to verify that the Laplace matrix corresponding to an undirected graph, which implies that $(v_j,v_i)\in{\cal E}$ if $(v_i,v_j)\in{\cal E}$ for $i,j\in\{1,2,\cdots,n\}$, is real symmetric.
\par For a leader of the MASs, those agents that can access it are considered as its neighbors. Define its out-degree matrix as $B=diag(b_{11},b_{11},\cdots,b_{nn})$. $b_{ii}=1$ if agent $i$ can access the leader, and $b_{ii}=0$ otherwise.
\begin{lemma}
	Let $H=L+B$. Then, $H$ is  positive definite. That is, $0<\eta_1(H)\le\eta_2(H)\le\cdots\le\eta_n(H)$, where $\eta_i(H)$ with $i\in\{1,2,\cdots,n\}$ is a characteristic value of $H$.
\end{lemma}
\par We have the following assumptions and definitions for the problem to be addressed in this article.
\begin{assumption}
	\label{assu1}
	Assume that the communication topology used in this article is undirected and connected.
\end{assumption}
\begin{assumption}
	\label{assu2}
	Assume that the agent managing ER can only be accessed by a small number of agents.
\end{assumption}
\begin{definition}
	\label{def1}
	In the grid connected mode, the neighbors of the agent managing ER are called flexible agents. Compared to other agents, flexible agents are only responsible for collecting the average power mismatch estimated by their neighboring agents, and then passing them on to the agent that manages ER. The agent that manages ER only receives information about average power mismatch estimation collected by its neighbors, and does not utilize MCs.
\end{definition}
\begin{assumption}
	\label{assu3}
	Assume that there is an agent that manages a BESS that is neither a flexible agent nor its neighbor. This indicates that not all agents are able to access the leader.
\end{assumption}
\section{Distributed PI+Reset schemes for discrete time economic dispatch of BESS network}
\par In this section, we will design a distributed PI+Reset scheme for a grid-connected BESS network with an ER, to accelerate convergence of MC and estimated power mismatch while minimizing operating costs and meeting the balance between supply and demand.
\begin{figure}
	\centering
	\includegraphics[width=8.5cm]{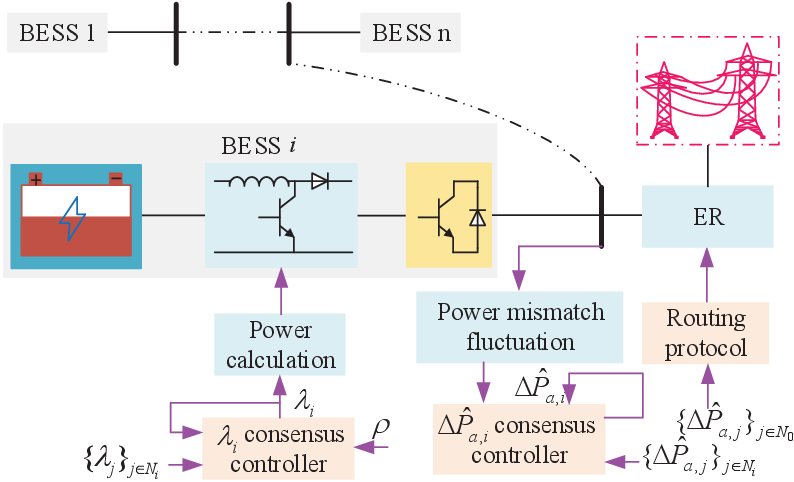} \caption{The control framework with a distributed ED scheme for a BESS network.\label{figei2}}
\end{figure}
\subsection{Distributed Control Scheme for MC Consensus}
\par Firstly, a positional PI controller based ED scheme for MCs is as follows
\begin{equation}
	\label{PIc}
	\lambda_i^{k+1}=\lambda_i^k-h_1\xi_{\lambda,i}^k-h_2\nu_{\lambda,i}^k,
\end{equation}
where $\xi_{\lambda,i}^k=\sum\limits_{i=1}^{n}a_{ij}(\lambda_i^k-\lambda_j^k)+b_{ii}(\lambda_i^k-\rho)$, the integral term is represented as the accumulation of errors, i.e., $\nu_{\lambda,i}^k=\sum\limits_{t=0}^{k} \xi_{\lambda,i}^t$, $h_1$ and $h_2$ are both positive numbers to be designed. This solution is capable of completing tasks and has high control accuracy. However, this can cause overshoot.
\begin{figure}
	\centering
	\includegraphics[width=7cm]{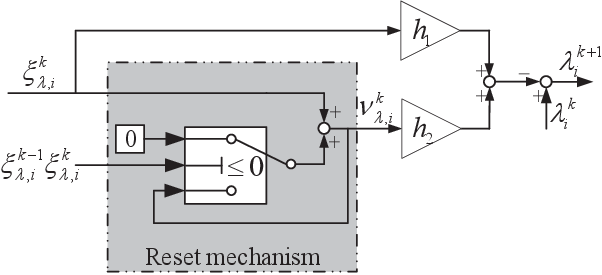} \caption{A PI controller with a error-dependent reset mechanism.\label{figpir}}
\end{figure}
\par To break this deadlock, we propose a state dependent switching controller as shown in Fig. \ref{figpir} based on \eqref{PIc}, i.e., a reset controller. Specifically, when the proportional term $\xi_{\lambda,i}^k$ changes sign, the integral term $\nu_{\lambda,i}^k$ is the term where the error accumulates again starting from 0. That is,
\begin{equation}
	\label{res}
	\nu_{\lambda,i}^k=\left\{\begin{aligned}
		\sum\limits_{t=l_i}^{k} \xi_{\lambda,i}^t,\; if\; \xi_{\lambda,i}^{k-1}\xi_{\lambda,i}^k> 0\\
		\xi_i^k,\; if\; \xi_{\lambda,i}^{k-1}\xi_{\lambda,i}^k\le 0
	\end{aligned}\right..
\end{equation}
\par Denote a jump instant set $l_\lambda$ as $l_\lambda:\{k|\xi_{\lambda,i}^{k-1}\xi_{\lambda,i}^k\le 0, i\in\{1,2,\cdots,n\}\}$. Define some stack vectors as $\xi_\lambda^k=[\xi^k_{\lambda,1},\xi^k_{\lambda,2},\cdots,\xi^k_{\lambda,n}]^T$, $\nu^k_\lambda=[\nu^k_{\lambda,1},\nu^k_{\lambda,2},\cdots,\nu^k_{\lambda,n}]^T$, and $\lambda^k=[\lambda^k_1,\lambda^k_2,\cdots,\lambda^k_n]$. Define a jump set and a flow set as ${\cal J}_\lambda:\{k|k\in l_\lambda\}$ and ${\cal F}_\lambda:\{k|k\notin l_\lambda\}$, respectively. Thus, we have
\begin{subequations}
	\label{6}
	\begin{equation}
		\xi^k_\lambda=H(\lambda^k-\rho),
	\end{equation}
	\begin{equation}
		\label{res1}
		\nu^k_\lambda=\left\{\begin{aligned}
			\nu^{k-1}_\lambda+\xi^{k-1}_\lambda,\; if\; k\in{\cal F}_\lambda\\
			\nu^{(k-1)+}_\lambda+\xi^{k-1}_\lambda,\; if\; k\in{\cal J}_\lambda
		\end{aligned}\right.,
	\end{equation}
	\begin{equation}
		\label{gl}
		\lambda^{k+1}=\lambda^k-h_1\xi^k_\lambda-h_2\nu^k_\lambda,
	\end{equation}
\end{subequations}
where $\nu^{(k-1)+}_\lambda$ is a vector obtained by resetting a certain dimension of $\nu^{k-1}_\lambda$.
\par For the grid connected mode, the base system \cite{banosResetControlSystems2011} is represented as
\begin{equation}
	\label{gs}
	x^{k+1}_\lambda=\Phi(H) x^k_\lambda
\end{equation}
where $x^k_\lambda=[(\xi^k_\lambda)^T\quad (\nu^k_\lambda)^T]^T$, $\Phi(H)=\left[\begin{matrix}
	I-h_1H & -h_2H\\
	I & I
\end{matrix}\right]$. And the characteristic polynomial of $\Phi(H)$ is
$$\begin{aligned}
	&\lvert \mu I-\Phi(H)\rvert\\
	=&\prod_{i=1}^n {\lvert\mu_i^2+(h_1\eta_i(H)-2)\mu_i+1+\eta_i(H)(h_2-h_1)\rvert}.
\end{aligned}$$
Thus, the eigenvalues of $\Phi$ can be obtained
\begin{equation}
	\mu_i=\frac{2-h_1\eta_i(H)\pm \sqrt{h_1^2\eta_i^2(H)-4h_2\eta_i(H)}}{2}.
\end{equation}
Denote $\Delta_i=h_1^2\eta_i^2(H)-4h_2\eta_i(H)$.
\begin{lemma}
	\label{eigl}
	Under the Assumption \ref{assu1}, if $\frac{4h_2}{h_1^2}\le\eta_m$ with $h_1,h_2>0$, or $\frac{4h_2}{h_1^2}>\eta_m$ with $0<h_2<h_1$, the base system \eqref{gs} is asymptotically stable, which leads to the system \eqref{6} being regular \cite{banosResetControlSystems2011}. Furthermore, matrix $\Phi$ has $2n$ real eigenvalues if $\frac{4h_2}{h_1^2}\le\eta_m$, while matrix $\Phi$ has at least a pair of conjugate eigenvalues if $\frac{4h_2}{h_1^2}>\eta_m$.
\end{lemma}
\begin{proof}
	The analysis of eigenvalues will be divided into two aspects, namely real eigenvalues and complex eigenvalues.
	\par Case I: $\Phi$ has at least one pair of conjugate complex eigenvalues. That is, there exists $\Delta_i< 0$ for $\exists i\in\{1,2,\cdots,n\}$. 
	$\frac{4h_2}{h_1^2}>\eta_i(H).$
	And then, $\vert \mu_{i,j}\vert<1$ for $\forall j\in\{1,2\}$ if $0<h_2<h_1$. Therefore, a conservative condition, i.e., $\frac{4h_2}{h_1^2}>\eta_m$, can ensure that the matrix $\Phi$ has at least one pair of conjugate complex eigenvalues.
	\par Case II. $\Phi$ has at least a real eigenvalue. Then there must be $\eta_i(H)$ such that 
	$\frac{4h_2}{h_1^2}\le\eta_i(H),$
	for $\exists i\in\{1,2,\cdots,n\}$, which can lead to a more conservative condition and a more radical condition, i.e., $\frac{4h_2}{h_1^2}\le\eta_M$ and $\frac{4h_2}{h_1^2}\le\eta_m$. The former ensures the existence of real eigenvalues, while the latter ensures the absence of complex eigenvalues. Besides, if $2-h_1\eta_i(H)\ge 0$ and $2-h_1\eta_i(H)+\sqrt{\Delta_i}<2$, i.e., 
	$h_1\le \frac{2}{\eta_i(H)}$
	, $\lvert \mu_{i,j} \lvert<1$ for $j\in\{1,2\}$. If $2-h_1\eta_i(H)<0$ and $(2h_1-h_2)\eta_i(H)<4$, i.e., $h_1> \frac{2}{\eta_i(H)}$, $\lvert \mu_{i,j} \lvert<1$ for $j\in\{1,2\}$. Overall, $\lvert \mu_{i,j} \lvert<1$ for $j\in\{1,2\}$ if $\frac{4h_2}{h_1^2}\le\eta_i(H)$.
	\par Combining Case I and II, Lemma \ref{eigl} can be derived. $\hfill\blacksquare$
\end{proof}
\begin{theorem}
	\label{TH1g}
	Under Assumption \ref{assu1}, the reset control system \eqref{gs} with \eqref{res} is regular with at least one reset instant if $\frac{4h_2}{h_1^2}>\eta_m$ and $0<h_2<h_1$ for the grid connected mode. 
\end{theorem}
\begin{proof}
	According to Lemma \ref{eigl}, if $\frac{4h_2}{h_1^2}>\eta_m$ and $h_1>0$, $\Phi$ has at least one pair of conjugate complex eigenvalues. In this way, $\vert \mu_{ij}\vert<1$ for $\textcolor{blue}{\forall i\in\{1,2,\cdots,n\}}$ and $\forall j\in\{1,2\}$ if $0<h_2<h_1$, and $\xi_i$ encounters with at least one zero crossing instant for $\exists i\in\{1,2,\cdots,n\}$. Thus, Theorem \ref{TH1g} has been proven. $\hfill\blacksquare$
\end{proof}
\textcolor{blue}{\begin{remark}
	From Lemma \ref{eigl}, it can be seen that compared to a PI controller, the stability of the system can not be does affected by the introduction of a reset mechanism because the eigenvalues of the base system can not be changed. However, the dynamic and steady-state performance of the system can be improved. Due to the fact that the integral term always has the same sign as the proportional term, the introduction of a reset mechanism results in the control input always being no less than the PI controller, especially when the base system has complex eigenvalues. Besides, it is ensured by Theorem \ref{TH1g} that the integral term $\nu_{\lambda,i}$ of each controller accumulates errors again when the proportional term $\xi_{\lambda,i}$ changes sign. This ensures that $\vert h_1\xi_{\lambda,i}+h_2\nu_{\lambda,i}\vert\ge (h_1+h_2)\vert\xi_{\lambda,i}\vert>h_1\vert\xi_{\lambda,i}\vert$, thereby accelerating the convergence rate. 
\end{remark}}
\begin{remark}
	\textcolor{blue}{Although the dynamic and steady-state performance of the system can be improved by the introduction of a reset mechanism,} this may lead to continuous resetting in time when the error is small. So, the reset mechanism can be improved to
	$$\nu_{\lambda,i}^k=\left\{\begin{aligned}
		\sum\limits_{t=l_i}^{k} \xi_{\lambda,i}^t,&\; if\; \xi_{\lambda,i}^{k-1}\xi_{\lambda,i}^k> 0\;\mathrm{and}\;\vert\xi_{\lambda,i}^k\vert>\epsilon\\
		\xi_i^k,&\; otherwise,
	\end{aligned}\right.$$
	where $\epsilon$ is a small positive constant.
\end{remark}
\par Next, we explain the relevant characteristics of the power mismatch estimation scheme.
\subsection{A Routing Algorithm Based on A Distributed Estimation Scheme for The Average Power Mismatch}
\par For flexible agents, their estimated average power mismatch are always 0 \cite{chenDistributedEconomicDispatch2021}, i.e.,
\begin{equation}
	\label{mps1}
	\Delta {\hat P}_{a,i}^{k+1}=0,
\end{equation}
where $\Delta {\hat P}_{a,i}^0=0$. For those who are not flexible agents, the estimated average power mismatch with reset mechanism is as follows
\begin{equation}
	\label{mps}
	\Delta\hat P_{a,i}^{k+1}=\Delta\hat P_{a,i}^k-z_1\xi^k_{P,i}-z_2\nu^k_{P,i}+\Delta P_i^{k+1}-\Delta P_i^k,
\end{equation}
where $\xi^k_{P,i}=\sum\limits_{i=1}^{n}a_{ij}(\Delta\hat P_{a,i}^k-\Delta\hat P_{a,j}^k)$, $z_1$ and $z_2$ are both positive constant to be designed, and the reset mechanism is
\begin{equation}
	\label{rems}
	\nu_{P,i}^{k+1}=\left\{\begin{aligned}
		\sum\limits_{t=K_i^l}^{k} \xi_{P,i}^t,\; if\; \xi_{P,i}^{k+1}\xi_{P,i}^k> 0\\
		\xi_{P,i}^k,\; if\; \xi_{P,i}^{k+1}\xi_{P,i}^k\le 0,
	\end{aligned}\right.
\end{equation}
where $\Delta {\hat P}_{a,i}^0=\Delta P_i^0$.
\par Here, we provide a modification method for the Laplace matrix corresponding to the original topology to adapt to the designed distributed control scheme with flexible agents. According to \cite{Gambuzza2021}, there must be a permutation matrix $J$, such that $J^TLJ=\left[\begin{matrix}
	L_1 & {\cal B}_1^T\\
	{\cal B}_1 & L_2
\end{matrix}\right]$, which distinguishes flexible agents from other agents on the Laplace matrix. Here, $L_1$ and $L_2$, respectively, correspond to the set of flexible agents and the set of other agents. At this point, the sum of any row in $L_1$ may not necessarily be 0. Thus, referring to the conclusion in \cite{Gambuzza2021}, there exists a Laplacian matrix ${\cal B}$ such that 
$J^T(L+{\cal B})J=\left[\begin{matrix}
	{\cal L}_1 & \mathbf{0}^T\\
	{\cal B}_1 & {\cal L}_2
\end{matrix}\right]$, which directly leads to $L={\cal L}+{\cal B}$ and ${\cal L}_2=L_2$, where ${\cal B}=(J^T)^{-1}\left[\begin{matrix}
L_1-{\cal L}_1 & \mathbf{0}^T\\
{\cal B}_1 & 0_{n-m}
\end{matrix}\right]J^T$. At this point, the sum of any row in ${\cal L}_1$ must be 0.
\begin{property}
	\label{p1}
	Under Assumption \ref{assu1}, the matrix $\cal L$ and $\cal B$ have the following properties.
	\begin{enumerate}
		\item Assume that $m_1$ and $m_2$, respectively, denote the number of zero eigenvalues of ${\cal L}_1$ and $L_2$. The matrix ${\cal L}$ has $m_1+m_2$ zero and $n-(m_1+m_2)$ positive eigenvalues.
		\item $-1^T_n{\cal L}=1^T_n({\cal B}^T-L)=1^T_n{\cal B}^T.$
		\item $1^T_n({\cal B}-I){\cal L}=-1^T{\cal L}.$
	\end{enumerate}
\end{property}
\begin{proof}
	1) in Property \ref{p1} is obvious because ${\cal L}_2$ is a lower triangular block matrix and ${\cal L}_2=L_2$.
	\par Due to $L={\cal L}+{\cal B}^T$ and $1^T_nL=0_n$, we have
	$$-1^T_n{\cal L}=1^T_n({\cal B}^T-L)=1^T_n{\cal B}^T,$$
	which implies 2) holds.
	\par Combining 2) and $1^T_n{\cal L}^T=0_n$, we have
	$$1^T_n({\cal B}-I){\cal L}=1^T_n(L^T-{\cal L}^T-I){\cal L}=-1^T{\cal L}.$$
	\par At this point, the proof is complete.$\hfill\blacksquare$
\end{proof}
\par Thus, a compact form of the power mismatch estimation scheme, i.e., \eqref{mps1} and \eqref{mps} with \eqref{rems},  can be written as
\begin{subequations}
	\label{13}
	\begin{equation}
		\xi_P^k={\cal L}\Delta\hat P_a,
	\end{equation}
	\begin{equation}
		\label{powermr0}
		\nu_P^k=\left\{\begin{aligned}
			\nu_P^{k-1}+\xi_P^{k-1},\; if\; k\in{\cal F}_P\\
			\nu_P^{(k-1)+}+\xi_P^{k-1},\; if\; k\in{\cal J}_P
		\end{aligned}\right.,
	\end{equation}
	\begin{equation}
		\label{powerm0}
		\Delta\hat P_a^{k+1}=\Delta\hat P_a^k-z_1\xi_P^k-z_2\nu_P^k+(I-B)(\Delta P^{k+1}-\Delta P^k),
	\end{equation}
\end{subequations}
where $l_P:\{k|\xi_{P,i}^{k-1}\xi_{P,i}^k\le 0, i\in\{1,2,\cdots,n\}\}$, ${\cal J}_P:\{k|k\in l_P\}$, and ${\cal F}_P:\{k|k\notin l_P\}$ are a reset instant set, a jump set and, a flow set, respectively, and $\Delta\hat P_a$, $\xi_P^k$, $\nu_P^k$, and $\Delta P^k$ are stack vectors of $\Delta\hat P_{a,i}$, $\xi_{P,i}^k$, $\nu_{P,i}^k$, and $\Delta P_i^k$, respectively.
\par Here is a routing algorithm based on Definition \ref{def1}. The total exchange power between a BESS network and UG, i.e. the output power of ER, is obtained from the following
\begin{equation}
	\label{PMGg}
	P_{UG}^{k+1}=\left\{\begin{aligned}
		P_{UG}^k-z_11^T_n{\cal B}^T\Delta \hat P^k_a-z_2\sum_{j=l_i}^{k}1^T_n{\cal B}^T\Delta \hat P^j_a\\
		+1_n^TB(\Delta P^{k+1}-\Delta P^k),\\
		if\;k\in{\cal F}_P\;\mathrm{and}\;l_i\le k\le l_{i+1},\\
		P_{UG}^k-z_11^T_n{\cal B}^T\Delta \hat P^k_a
		+1_n^TB(\Delta P^{k+1}-\Delta P^k),\\
		if\;k\in{\cal J}_P\;\mathrm{and}\;k=l_i,
	\end{aligned}
	\right.
\end{equation}
where $P_{UG}^k$ denotes the power exchange between UG and BESSs, and $P_{UG}^0=1_n^T(\Delta P^0-\Delta {\hat P}_a^0)$. ${\cal B}^T\Delta \hat P^j_a$ denotes the collected estimated average power mismatch by the neighbors of ER, according to Definition \ref{def1}.
\par The base system of \eqref{powerm0} with \eqref{powermr0} is represented as
\begin{equation}
	\label{gms}
	x_P^{k+1}=\Phi({\cal L}) x_P^k
\end{equation}
where $x_P^k=[(\xi_P^k)^T\quad (\nu_P^k)^T]^T$, $\Phi({\cal L})=\left[\begin{matrix}
	I-z_1{\cal L} & -z_2{\cal L}\\
	I & I
\end{matrix}\right]$.
\par Referring to Property \ref{p1}, the system \eqref{gms} with \eqref{rems} follows a conclusion similar to  Theorem \ref{TH1g}. So, the following theorem is given without proof.
\begin{theorem}
	\label{TH2}
	Under the Assumption \ref{assu1}, if $z_1>0$ and $z_2>0$, the base system \eqref{gms} is asymptotically stable. Furthermore, matrix $\Phi^d$ has $n$ real eigenvalues if $\frac{4z_2}{z_1^2}\le\eta_{m_1+m_2+1}({\cal L})$, and the matrix $\Phi^d$ has at least a pair of conjugate complex eigenvalues if $\frac{4z_2}{z_1^2}>\eta_{m_1+m_2+1}({\cal L})$ with $0<z_2<z_1$, resulting that the base system \eqref{gms} with the reset mechanism \eqref{rems} is regular with at least one reset instant.
\end{theorem}
\subsection{Supply and Demand Balance Analysis}
\par Define two error vectors $\delta_\lambda^k$ and $\delta_P^k$ as $\delta_\lambda^k=\lambda^k-\rho 1$ and $\delta_P^k=\Delta {\hat P}^k-0_n$, respectively. Then, referring \eqref{gl} and \eqref{mps}, we have
\begin{subequations}
	\begin{equation}
		\label{deltal}
		\delta_\lambda^k=\delta_\lambda^{k-1}-h_1\xi_\lambda^{k-1}-h_2\nu_\lambda^{k-1},
	\end{equation}
	\begin{equation}
		\label{deltap}
		\delta_P^k=\delta_P^{k-1}-h_1\xi_P^{k-1}-h_2\nu_P^{k-1}+(I-B)(\Delta P^{k+1}-\Delta P^k).
	\end{equation}
\end{subequations}
The compact form of \eqref{deltal} and \eqref{deltap} are written as
\begin{subequations}
	\begin{equation}
		\label{deltal1}
		y^{k+1}_\lambda=\Phi({\cal L})y^k_\lambda,
	\end{equation}
	\begin{equation}
		\label{deltp1}
		y^{k+1}_P=\Phi({\cal L})y^k_P+\left[\begin{matrix}
			(I-B)(\Delta P^{k+1}-\Delta P^k)\\
			0
		\end{matrix}\right].
	\end{equation}
\end{subequations}
where $y_\lambda^k=col(\delta_\lambda^k,\nu_\lambda^k)$, $y_P^k=col(\delta_P^k,\nu_P^k)$. For \eqref{deltal1} and \eqref{deltp1}, convergence analysis is summarized in Theorem \ref{TH3} and relevant proof is provided.
\begin{theorem}
	\label{TH3}
	Based on the results of Theorem \ref{TH1g}, MC asymptotically converges to the real-time electric price with at least one reset instant. Based on the results of Theorem \ref{TH2} and $\nu_P^{l_{P,i}}=0_n$ for $l_{P,i}\in l_P$, the estimated average power mismatch asymptotically converges to 0 with at least one reset instant. Furthermore, the supply-demand balance can be asymptotically guaranteed by $1^T[\Delta {\hat P}^0_a-\Delta P^0]+P_{UG}^0=0$.
\end{theorem}
\begin{proof}
	The proof of Theorem \ref{TH3} is divided into three parts. In \emph{Part A}, it is responsible for explaining the stability and convergence of MC. In \emph{Part B}, an auxiliary property equation, $1^T[\Delta {\hat P}^{k+1}_a-\Delta P^{k+1}]+P_{UG}^{k+1}=0$, is proven to hold at all times. Finally, in \emph{Part C}, the stability and convergence of the estimated power mismatch are explained.
	\par \emph{Part A}. Before the first rest instant, \eqref{deltal1} can be deduced as
	$$y^k_\lambda=\Phi^k(H)y^0_\lambda.$$
	Before the second reset instant, \eqref{deltal1} can be deduced as
	$$y^k_\lambda=\Phi^{k-l_1}(H)y^{l_1+}_\lambda.$$
	Due to $\Vert y^{l_1+}_\lambda\Vert\le\Vert y^{l_1}_\lambda\Vert$, it is evident that
	\begin{equation}
		\label{deltan}
		\begin{array}{l}
			\Vert y^k_\lambda\Vert=\Vert\Phi^{k-l_1}(H)y^{l_1+}_\lambda\Vert\\
			\le\Vert\Phi^{k-l_1}(H)\Vert\Vert y^{l_1+}_\lambda\Vert\le\Vert\Phi^{k}(H)\Vert\Vert y^{0}_\lambda\Vert,
		\end{array}
	\end{equation}
	which leads to
	$$\lim\limits_{k\to+\infty}\Vert y^k_\lambda\Vert\le\lim\limits_{k\to+\infty}\Vert\Phi^{k}(H)\Vert\Vert y^{0}_\lambda\Vert=0.$$
	The above conclusion gives $\lim\limits_{k\to+\infty} y^k_{\lambda,i}=0$, i.e., $\lim\limits_{k\to+\infty}\delta^k_{\lambda,i}=0$ and $\lim\limits_{k\to+\infty}\nu_{\lambda,i}^k=0$ for $i\in\{1,2,\cdots,n\}$. This means that $\Delta P_i^k$ for $i\in\{1,2,\cdots,n\}$ derived from \eqref{power} converges, i.e., $\lim\limits_{k\to+\infty}(\Delta P_{a,i}^{k+1}-\Delta P_{a,i}^k)=0$.
	\par \emph{Part B}. 
	With the help of the above, $\nu_P^{l_{P,i}}=0_n$ for $\forall l_{P,i}\in l_P$, and Property \ref{p1}, for $l_{P,i}<k<l_{P,i+1}$, \eqref{13} and \eqref{PMGg} conform to the following derivation, i.e.,
	$$\begin{array}{l}
		1^T[\Delta {\hat P}^{k+1}_a-\Delta P^{k+1}]+P_{UG}^{k+1}\\
		=1^T[(I-z_1{\cal L})\Delta {\hat P}^k_a-z_2v_P^k-\Delta P^k]+P_{UG}^k\\
		-z_11^T_n{\cal B}^T\Delta \hat P^k_a-z_2\sum_{j=l_{P,i}}^{k}1^T_n{\cal B}^T\Delta \hat P^j_a\\
		=1^T[(I-z_1L)\Delta {\hat P}^k_a-z_2L\sum_{j=l_{P,i}+1}^{k}\Delta {\hat P}^k_a\textcolor{blue}{-\Delta P^k}]
		+P_{UG}^k\\
		=1^T[\Delta {\hat P}^k_a-\Delta P^k]+P_{UG}^k\\
		=\cdots\\
		=1^T[\Delta {\hat P}^{l_{P,i}+1}_a-\Delta P^{l_{P,i}+1}]+P_{UG}^{l_{P,i}+1}\\
		=1^T[(I-z_1L)\Delta {\hat P}^{l_{P,i}}_a-\Delta P^{l_{P,i}}]+P_{UG}^{l_{P,i}}\\
		=1^T[\Delta {\hat P}^{l_{P,i}}_a-\Delta P^{l_{P,i}}]+P_{UG}^{l_{P,i}}\\
		=\cdots\\
		=1^T[\Delta {\hat P}^0_a-\Delta P^0]+P_{UG}^0.
	\end{array}$$	
	According to the given condition $1^T[\Delta {\hat P}^0_a-\Delta P^0]+P_{UG}^0=0$, 
	\begin{equation}
		\label{balance1}
		1^T[\Delta {\hat P}^{k+1}_a-\Delta P^{k+1}]+P_{UG}^{k+1}=0,\;k=\{0,1,2,\cdots\}.
	\end{equation}
	\par \emph{Part C}. Before the first reset instant,
	$$\begin{array}{l}
		y^k_P=\Phi^k({\cal L})y^0_P\\
		+\sum_{j=0}^{k-1}\Phi^{k-1-j}({\cal L})\left[\begin{matrix}
			(I-B)(\Delta P^{k+1}-\Delta P^k)\\
			0
		\end{matrix}\right].
	\end{array}$$
	After that, before any reset instant, we have
	$$\begin{array}{l}
		y^k_P=\Phi^{k-l_1}({\cal L})y^{l_1+}_P\\
		+\sum_{j=l_1}^{k-1}\Phi^{k-1-j}({\cal L})\left[\begin{matrix}
			\Delta P^{j+1}-\Delta P^j\\
			0
		\end{matrix}\right].
	\end{array}$$
	Considering $\Vert y^{l_1+}_P\Vert\le\Vert y^{l_1}_P\Vert$ and the convergence of $\Delta {\hat P}^k$, denote $\Vert\Delta P^{j+1}-\Delta P^j\Vert\le \Delta \bar P$ with $\lim\limits_{j\to+\infty}\Delta{\bar P}^j$=0. Then, one can get
	$$\begin{array}{l}
		\Vert y^k_P\Vert\le\Vert \Phi_d^ky^{l_1+}_P\Vert\\
		+\Vert \sum_{j=l_1}^{k-1}\Phi_d^{k-1-j}\left[\begin{matrix}
			(I-B)(\Delta P^{k+1}-\Delta P^k)\\
			0
		\end{matrix}\right]\Vert\\
		\le \alpha^{k-l_{P,1}}\Vert y^{l_{P,1}+}_P\Vert+\sum_{j=l_{P,1}}^{k-1}\alpha^{k-1-j}\Delta\bar{P}^{j}\\
		\le \alpha^{k-l_{P,1}}\Vert y^{l_{P,1}}_P\Vert+\sum_{j=l_{P,1}}^{k-1}\alpha^{k-1-j}\Delta\bar{P}^{j}\\
		\le \alpha^k\Vert y^0_P\Vert+\sum_{j=0}^{k-1}\alpha^{k-1-j}\Delta\bar{P}^{j}.\\
	\end{array}$$
	So, at any instant, it can be asserted that
		\begin{equation}
		\label{19}
		\begin{array}{l}
			\Vert y^k_P\Vert\le \alpha^k\Vert y^0_P\Vert+\sum_{j=0}^{k-1}\alpha^{k-1-j}\Delta\bar{P}^{j}\\
			=o(1),k\to+\infty,
		\end{array}
	\end{equation}
	which gives
	$$\lim\limits_{k\to+\infty}\Delta {\hat P}^k_{a,i}=0,\;\lim\limits_{k\to+\infty}\upsilon^k_i=0,\;i\in\{1,2,\cdots,n\}.$$
	Based on \eqref{balance1} and the above, we dare to assert that
	$$\lim\limits_{k\to+\infty}P_{UG}^k=\lim\limits_{k\to+\infty}1^T\Delta P^k,$$
	which indicates that the supply-demand balance is asymptotically guaranteed. 
	\par At this point, Theorem \ref{TH3} has been proven. $\hfill\blacksquare$
\end{proof}
\begin{remark}	
	\label{pro}
	Here, we explain the progressiveness of the PI+Reset controller in control accuracy and convergence speed. Taking MC as an example, the following transformation is made for \eqref{res}, i.e., $\delta h_i^k=\frac{\nu_{\lambda,i}^k}{\xi_{\lambda,i}^k},$
	which leads to $\lambda_i^{k+1}=\lambda_i^k-(h_1+h_2\delta h_i^k)\xi_{\lambda,i}^k$. The reset mechanism shown in \eqref{res} results in, for $k\in\{1,2,\cdots\}$, the integral term $\nu_{\lambda,i}^k$ always maintaining the same sign as the proportional term $\xi_{\lambda,i}^k$, which ensures that $\delta h_i^k\ge 0$, thereby accelerating system convergence. The integral effect of the PI+Reset controller takes effect between two adjacent reset times, which also improves the control accuracy compared to the P controller. The above inference will be verified in subsequent simulations.
\end{remark}
\begin{remark}
	\label{reset}
	From Theorem \ref{TH3}, for the grid-connected node, it can be observed that at each reset instant, the average power mismatch estimation scheme requires a centralized reset operation to ensure power supply and demand balance, while the MC control scheme does not. However, a solution with a distributed reset mechanism may not necessarily perform better in dynamic performance than a centralized one. Subsequent simulations will validate this assertion.
\end{remark}
\par Therefore, there is a following proposition regarding solving the ED problem of a grid connected BESS network.
\begin{proposition}
	According to the result of Theorem \ref{TH3}, the ED problem of grid-connected BESSs can be solved by using \eqref{power}, \eqref{PIc} with \eqref{res}, and \eqref{mps} with \eqref{rems}.
\end{proposition}
\section{Some Simulation Cases}
\par In order to verify the effectiveness and progressiveness of the designed algorithm, five case studies are arranged. The adopted BESS network and its accompanying communication network are shown in Fig. \ref{fig3}, where agent 1, according to Definition \ref{def1}, is considered as a flexible agent. Thus, the matrices that need to be used are given as follow,
$$L=\left[\begin{matrix}
	3&-1&-1&-1\\
	-1&2&0&-1\\
	-1&0&2&-1\\
	-1&-1&-1&3
\end{matrix}\right],{\cal L}=\left[\begin{matrix}
0&0&0&0\\
-1&2&0&-1\\
-1&0&2&-1\\
-1&-1&-1&3
\end{matrix}\right],$$ 
$$B=\left[\begin{matrix}
	1&0&0&0\\
	0&0&0&0\\
	0&0&0&0\\
	0&0&0&0
\end{matrix}\right], {\cal B}=\left[\begin{matrix}
	3&0&0&0\\
	-1&0&0&0\\
	-1&0&0&0\\
	-1&0&0&0
\end{matrix}\right].$$
\par Here, six case studies are organized to fully test the design scheme
on effectiveness, progressiveness, gain change, 
time phased electricity price and plug and play.
\par In Case 1, the designed distributed solution with PI+Reset controllers is validated.
 And the simulation results are compared with the scheme in \cite{chenDistributedEconomicDispatch2021} to illustrate the progressiveness, as discussed in Remark\ref{pro}\textcolor{blue}{, where the relevant comparative indicators are shown in Table \ref{table}}.
\par In Case 2, the impact of the parameters of the PI+Reset
 controller on the simulation results is investigated.
\par In Case 3, the effectiveness of the designed scheme is tested in a phased electricity price environment.
\par In Case 4, the effect of the designed method on load switching is verified.
\par In response to SoC level limitation, a case considering plug and play is set up and tested for effectiveness in Case 5. A solution that can meet the plug and play function of BESSs is sufficient to handle the access and exit of BESS at any time under conditions such as SoC level limitations.
\par Apply the designed scheme to a large-scale power system in Case 6 and perform simulation to verify the performance of the designed scheme.
\begin{table}[h]
		\centering
		\textcolor{blue}{\caption{Comparison of Two Schemes on MC Consensus}}\label{table}%
		\textcolor{blue}{\begin{tabular}{@{}ccc@{}}
			\toprule
			\diagbox{Indicators}{Schemes}
			&\makecell{The one in this paper}&\makecell{The one in \cite{chenDistributedEconomicDispatch2021}}\\
			\midrule
			Settling Time &$25s$&$85s$\\
			\addlinespace
			Consensus Time &$25s$&$85s$\\
			\toprule
		\end{tabular}}
\end{table}
\begin{figure}
	\centering
	\includegraphics[width=8cm]{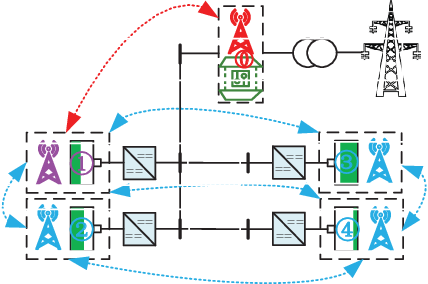} \caption{A BESS network with an ER and their communication graph.\label{fig3}}
\end{figure}
\begin{figure}
	\centering
	\includegraphics[width=8cm]{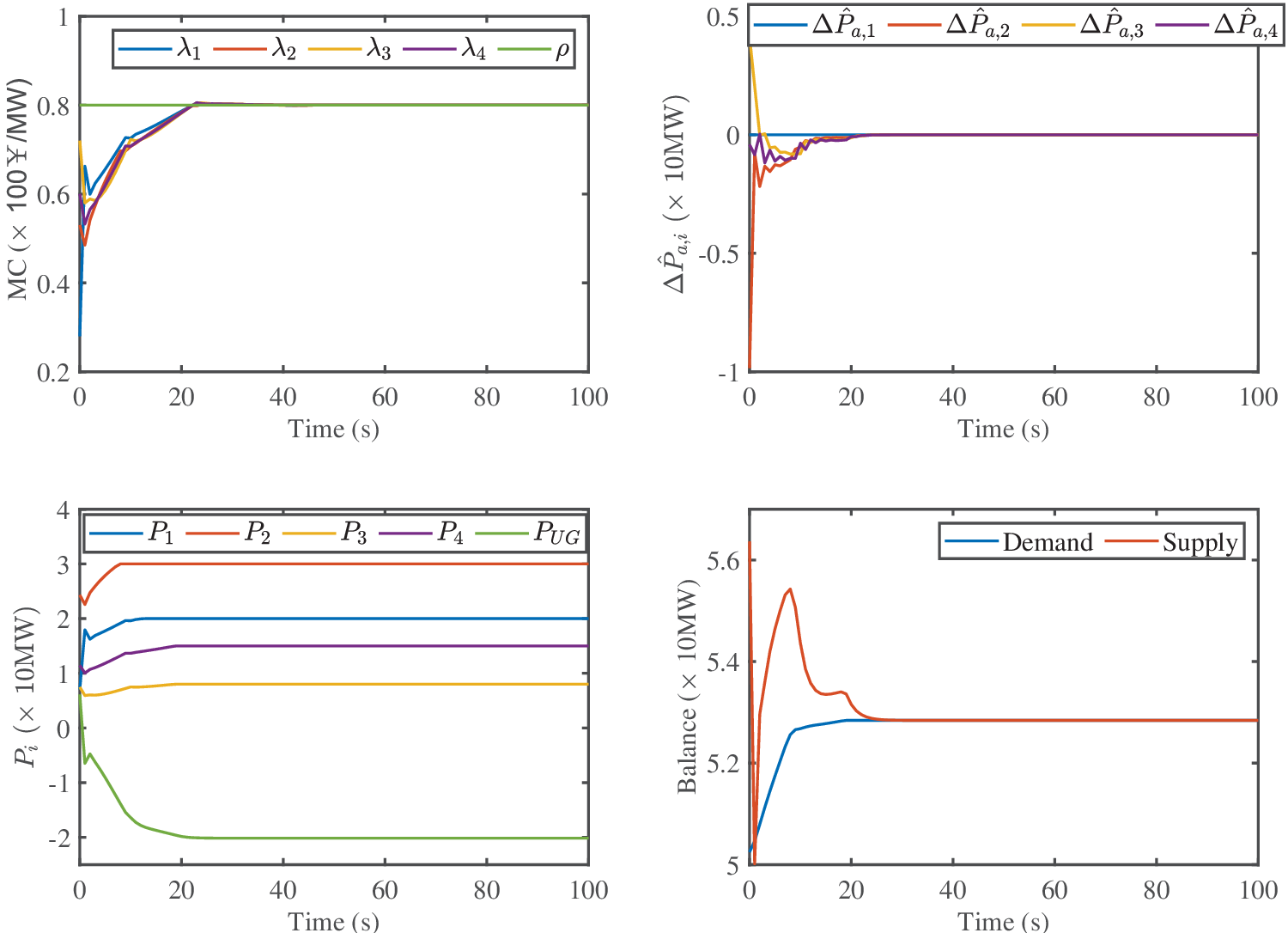} \caption{Simulation results of a distributed scheme with PI+Reset controllers in Case 1.\label{figc1}}
\end{figure}
\begin{figure}
	\centering
	\includegraphics[width=8cm]{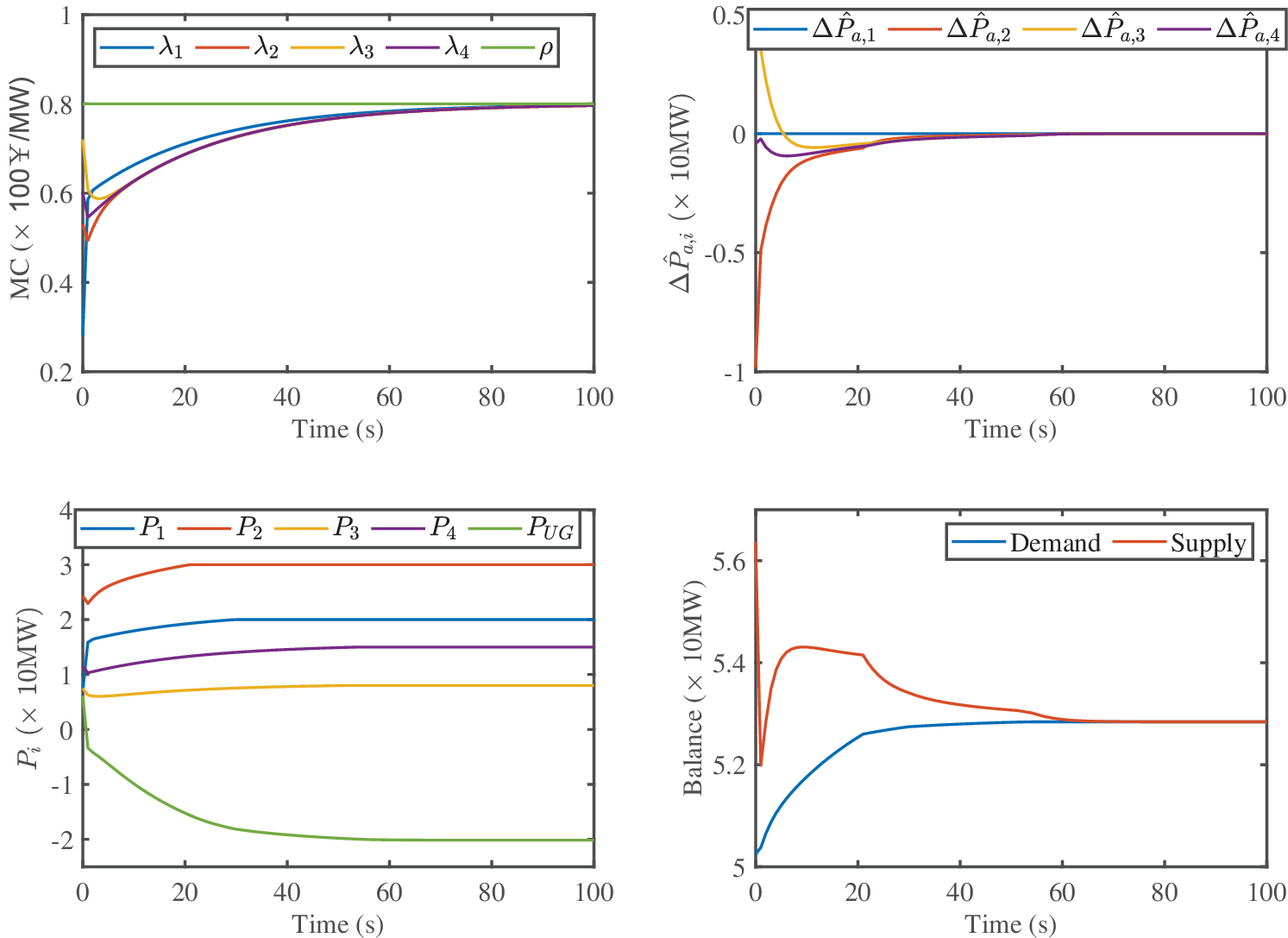} \caption{Simulation results of a distributed scheme designed in \cite{chenDistributedEconomicDispatch2021} in Case 1.\label{figc1c}}
\end{figure}
\begin{figure}
	\centering
	\includegraphics[width=8cm]{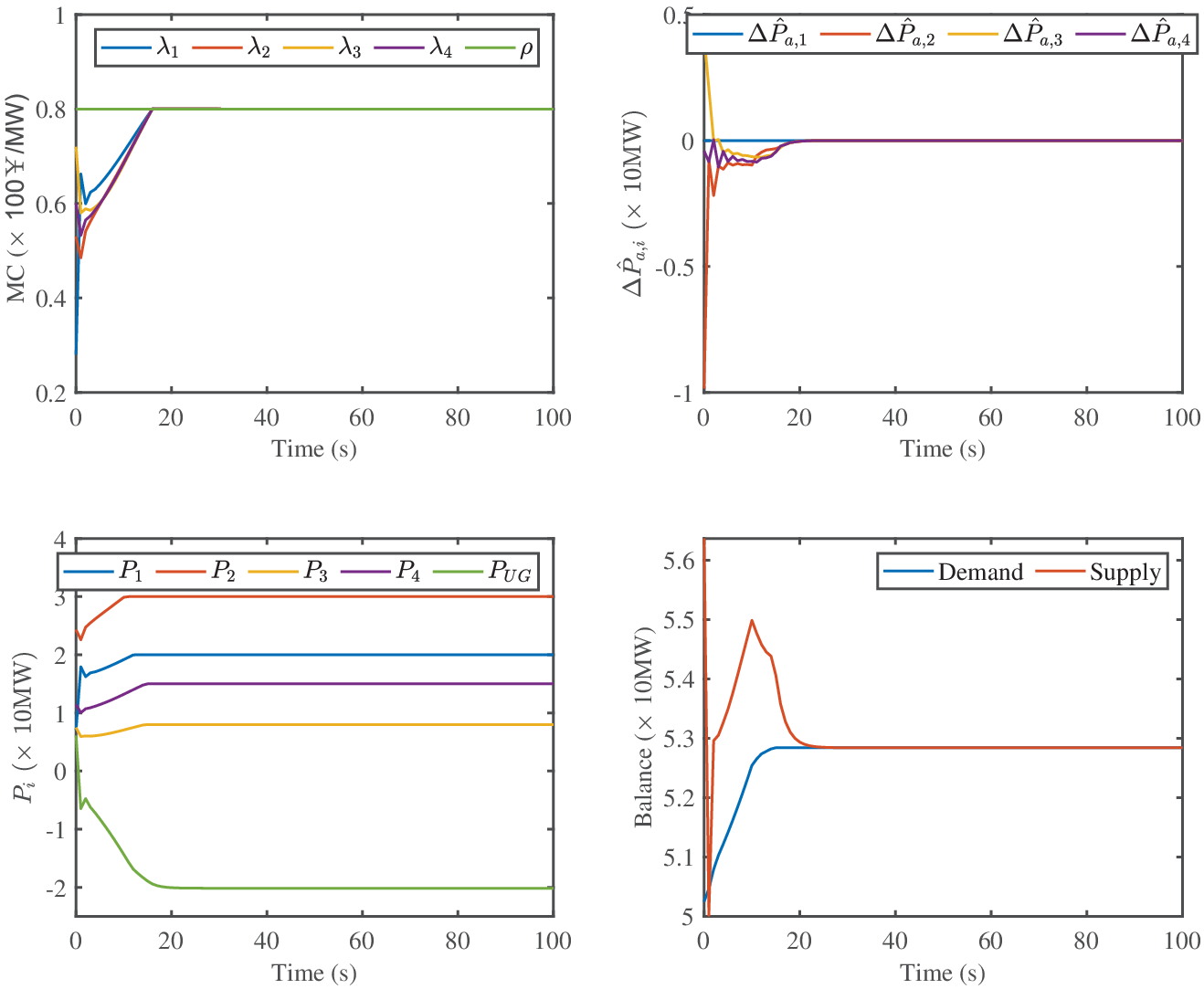} \caption{Simulation results of a distributed scheme with PI controllers with centralized reset mechanisms in Case 1.\label{figc1c1}}
\end{figure}
\subsection{Case 1. The Test on The Effectiveness of The Designed Scheme}
\par This case study is arranged to investigate the effect of the designed scheme to deal with the ED problem of the grid-connected BESS network shown in Fig. \ref{fig3}, and illustrate the progressiveness. The simulation results on MC, output power, average power mismatch estimation, and supply-demand balance are presented in Fig. \ref{figc1}. Correspondingly, the scheme designed in \cite{chenDistributedEconomicDispatch2021} is implemented for this BESS network for comparison, and the simulation results can be found in Fig. \ref{figc1c}.
\par From the simulation results in Fig. \ref{figc1}, it can be seen that MCs converge  asymptotically to the electricity price. The output power of each BESS is  asymptotically increasing and ultimately constrained by the maximum output power limit. On the contrary, the output power of UG  asymptotically decreases and ultimately remains stable, indicating that this BESS network is selling power to UG in exchange for profits. The estimation of average power mismatch ultimately converges to 0. Moreover, agent 1, as a flexible agent, can ensure that the estimated average power mismatch is always 0. And the balance between supply and demand can also be ensured.
\par Comparing the results of the two schemes, our scheme has significant advantages in convergence speed and control accuracy. This conclusion can be inferred by comparing the simulation results of any pair of variables. However, the designed reset scheme still causes some oscillations, as the current scheme cannot always keep $\xi^\lambda_i$ ($\xi^P_i$) and $\nu^\lambda_i$ ($\nu^P_i$) the same sign. It seems that there is still room for development here.
\par To verify Remark \ref{reset}, a centralized reset mechanism is introduced for the MC consensus scheme, i.e.,
$$\nu^k_\lambda=\left\{\begin{aligned}
	\nu^{k-1}_\lambda+\xi^{k-1}_\lambda,\; if\; k\in{\cal F}_\lambda\\
	\xi^{k-1}_\lambda,\; if\; k\in{\cal J}_\lambda.
\end{aligned}\right.$$
Based on this, the simulation results are shown in Fig. \ref{figc1c1}. Compared with Fig. \ref{figc1}, the solution with centralized reset mechanism exhibits better dynamic performance, and the reason for this is not yet clear, which is a problem worth exploring. When choosing a reset mechanism, there is a trade-off between a centralized reset mechanism and a decentralized one, which is because although the centralized reset mechanism leads to better results, it requires a control center. The following simulations adopt a decentralized reset mechanism unless otherwise specified.
\subsection{Case 2. The Test on The Impact of Gains on Simulation Results}
\par In order to investigate the impact of controller gains, i.e., $h_1$, $h_2$, $z_1$, and $z_2$, on the effectiveness of the designed scheme, three different values for each gain are selected to execute the simulation program. MC and $\Delta {\hat P}_{a,i}$ are selected for comparison under different gains in this case, as shown in Fig. \ref{figc2} and \ref{figc2c}.
\par From Fig. \ref{figc2}, it can be seen that with the increase of $h_1$, MC converges faster to the electricity price, and this rate has significantly improved. With the increase of $h_2$, the improvement in this rate is also significant, and the control accuracy is significantly increasing. However, a larger $h_2$ can cause local oscillations. Therefore, a larger $h_1$ is chosen to improve the convergence rate, while $h_2$ should not be set too large, according to Lemma \ref{eigl}, to ensure stability.
\begin{figure}
	\centering
	\includegraphics[width=8.4cm]{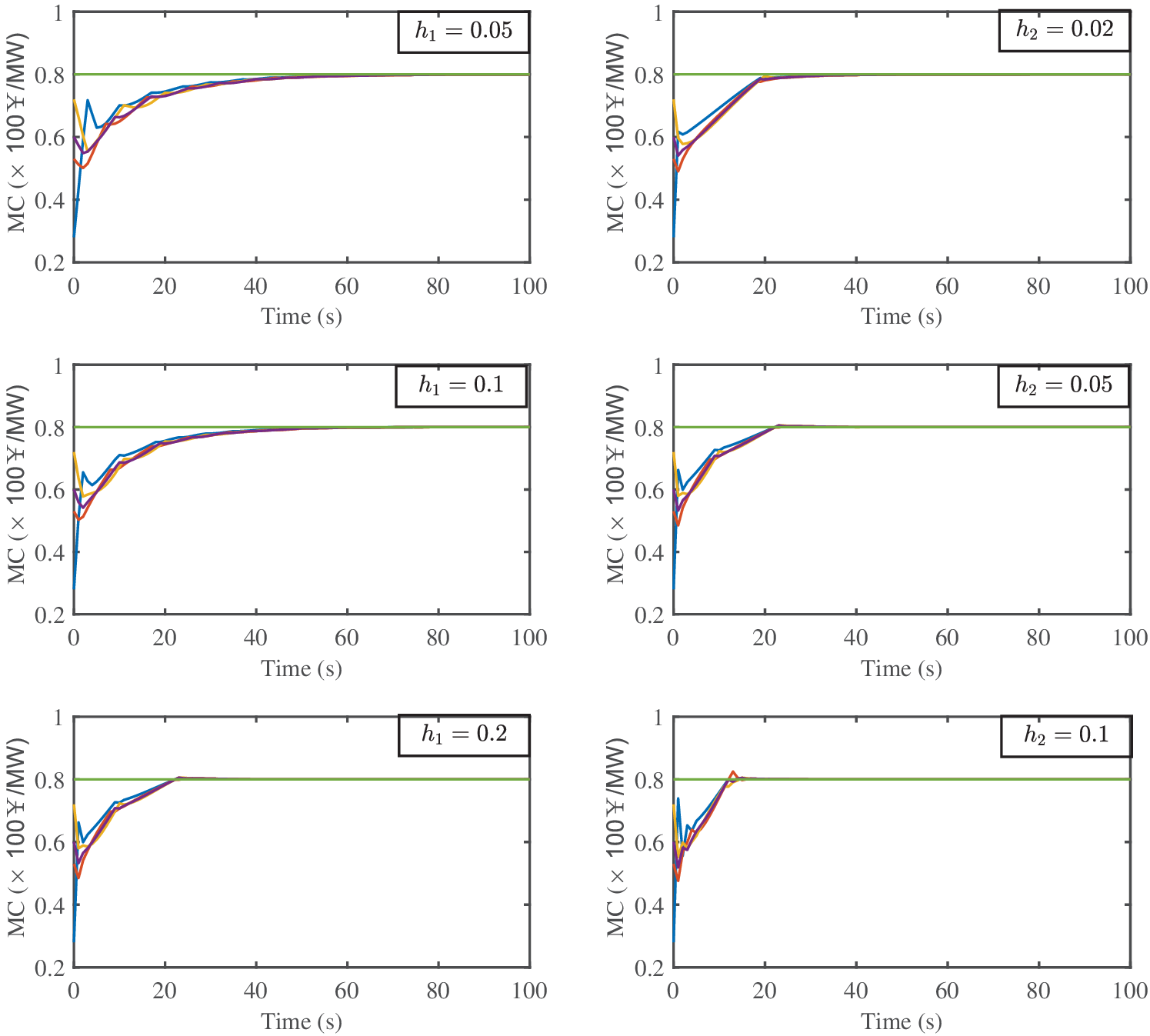} \caption{The impact of parameters $h_1$ and $h_2$ on MC consensus in Case 2.\label{figc2}}
\end{figure}
\par From Fig. \ref{figc2c}, it can be seen that with the increase of $z_1$ and $z_2$, estimated average power mismatch converge to 0 significantly faster.
\begin{figure}
	\centering
	\includegraphics[width=8.4cm]{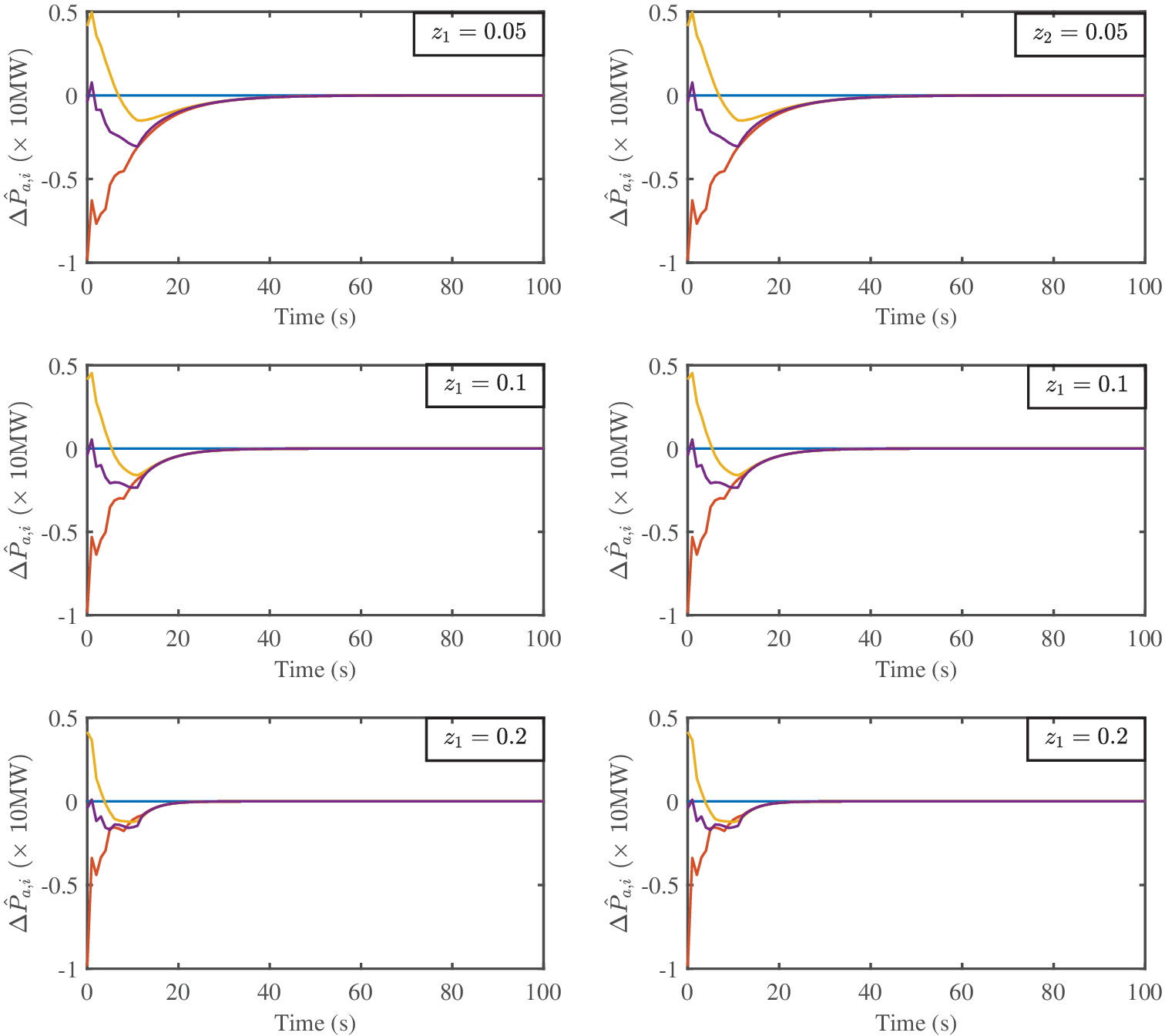} \caption{The impact of parameters $z_1$ and $z_2$ on average power mismatch estimation in Case 2.\label{figc2c}}
\end{figure}
\begin{figure}
	\centering
	\includegraphics[width=8cm]{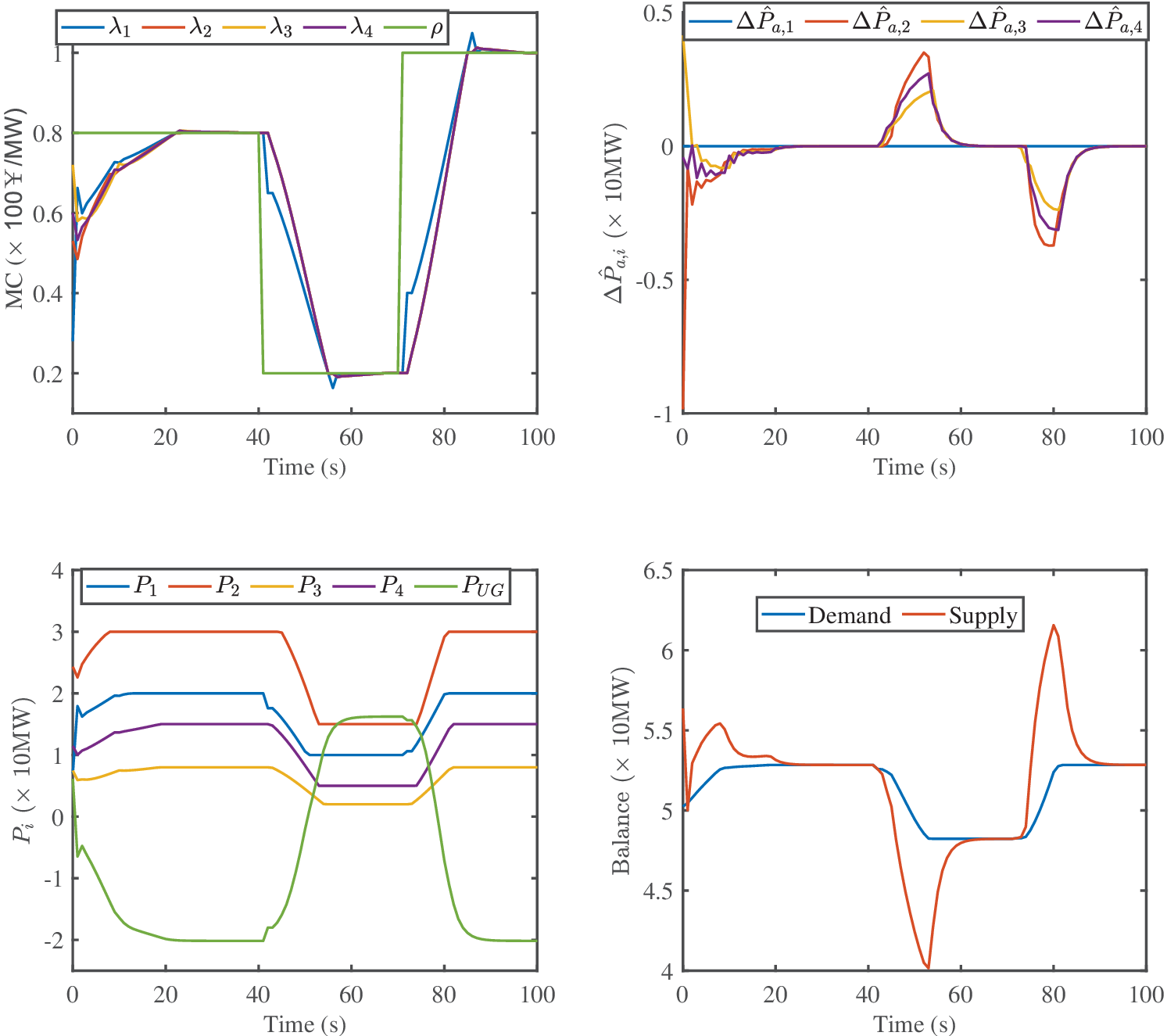} \caption{Simulation results of a distributed scheme with PI+Reset controllers under the time phased electricity price environment in Case 3.\label{figc3}}
\end{figure}
\subsection{Case 3. The Test in The Time Phased Electricity Price Environment}
\par On the basis of Case 1, this case is designed to investigate the performance of the designed algorithm in a time phased electricity price environment. The simulation results are shown in Fig. \ref{figc3}.
\par From the simulation results, it can be seen that in the time phased electricity price environment, MC can converge to the electricity price, which also leads to the stable BESS output
power and UG output power. Meanwhile, the average power mismatch estimation can converge to 0 and the supply-demand balance can be well maintained. These do not require resetting any variables. It should be noted that low electricity prices result in this BESS network absorbing power from UG. Due to output power limitations, increasing electricity prices becomes meaningless when the output power of all employees is at its maximum. Users can only reduce their electricity bills and overall expenses by expanding or introducing more BESSs. This is precisely the result of relying on electricity prices to guide users to use electricity reasonably.
\subsection{Case 4. The Test on Load Switching}
\begin{figure}
	\centering
	\includegraphics[width=8cm]{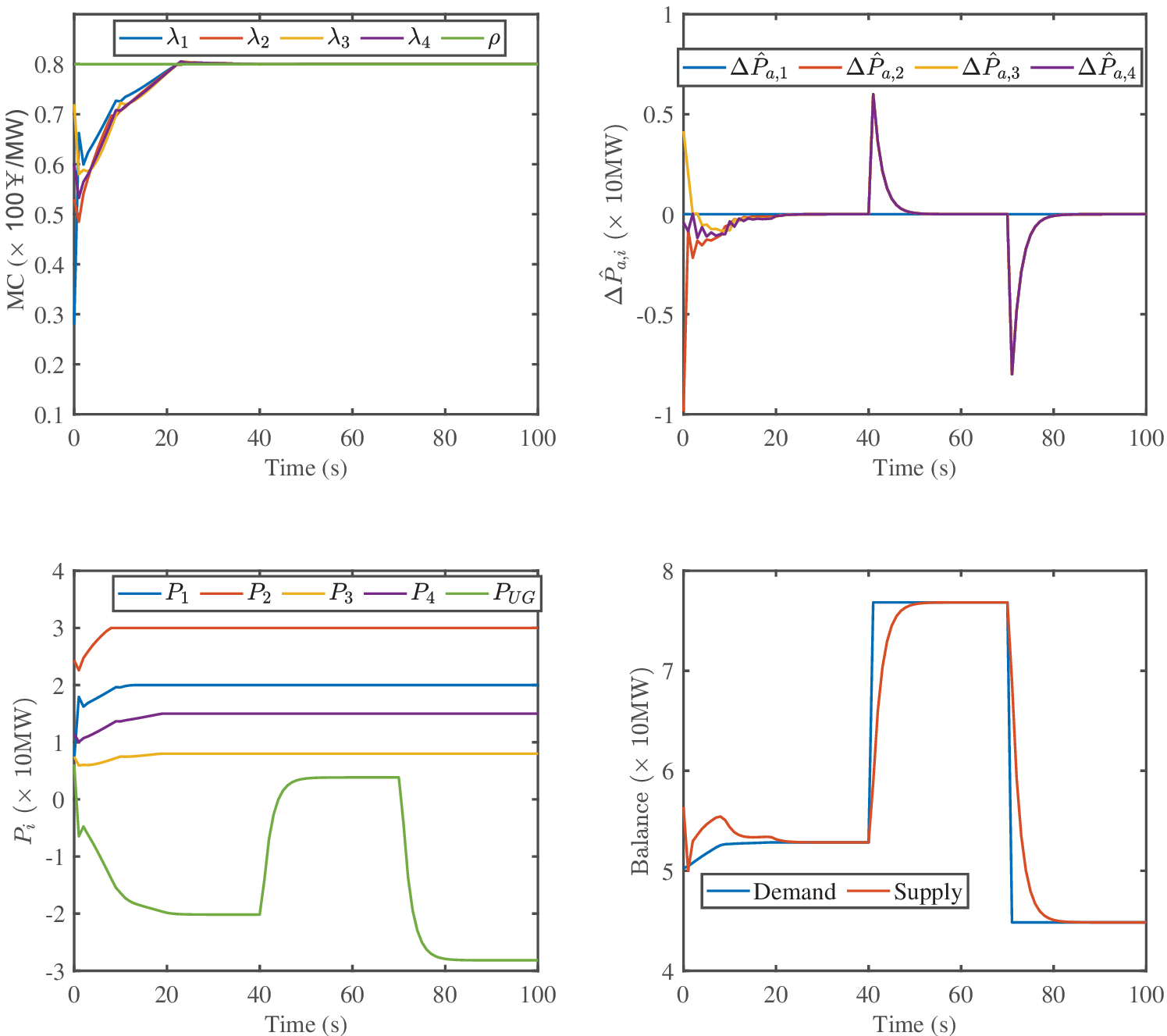} \caption{Simulation results on load switching in Case 4.\label{figc4}}
\end{figure}
\par On the basis of Case 1, this case is organized to test the designed ED scheme for a BESS network with load switching. The simulation results are shown in Fig. \ref{figc4}.
\par Since the change in load does not affect MC, the change in load will not be shared by any BESS, but will be borne solely by MG. That is to say, regardless of how the load changes, the output power of BESS remains unchanged, resulting in a constant line loss power. In addition, the average power mismatch estimation can converge to 0 under load switching conditions. The balance between power supply and demand can also be well maintained.
\begin{figure}
	\centering
	\includegraphics[width=8cm]{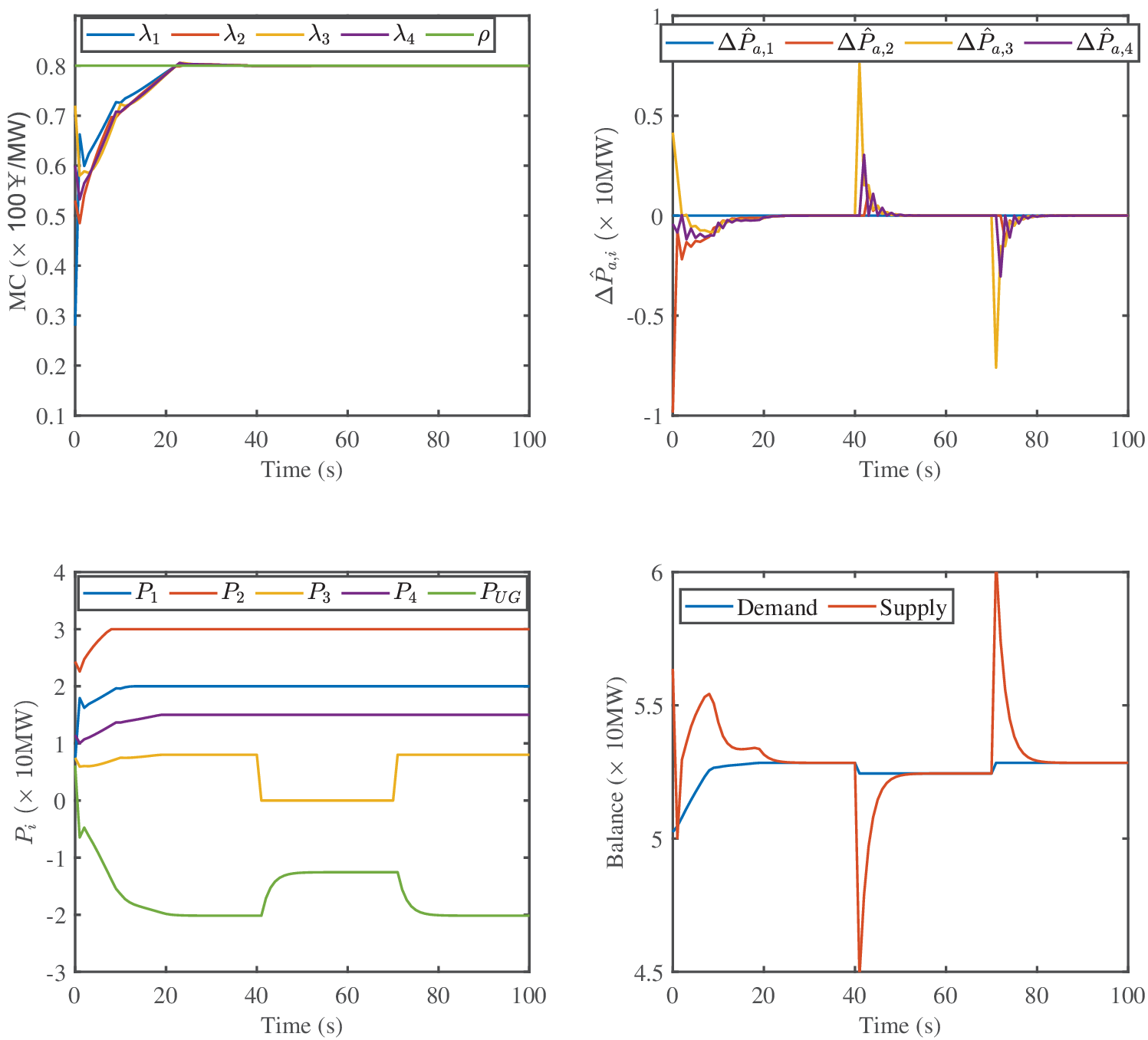} \caption{Simulation results on load switching in Case 5.\label{figc5}}
\end{figure}
\subsection{Case 5. The Test on The Plug and Play}
\par Due to the level of SoC, a BESS may exit at any time. Therefore, in this case, the designed ED scheme is tested in response to the access and exit of a BESS. At $ t=35s$, BESS 3 is removed. At this point, it is equivalent to reducing the capacity of BESS 3 to 0. At $t=70s$, BESS 3 resumes the connection. The simulation results are shown in Fig. \ref{figc5}.
\par From the simulation results, it can be seen that due to the removal of BESS 3, its output power has decreased to 0. This part of the power shortage is supplemented by UG until BESS 3 is reconnected, as other BESSs are operating at full load. The average power mismatch estimation of each BESS converges to 0 after experiencing some oscillations near the removal and restoration instants of BESS 3. Due to the absence of BESS 3, the total power loss has decreased, but this does not affect the balance of power supply and demand.
\par All of these indicate that as long as the communication topology is connected, our designed solution can solve the ED problem of a BESS network, whether accessing or removing some BESSs. This plug and play solution is also effective in addressing the ED problem of a BESS network constrained by SoC levels.
\begin{figure}
	\centering
	\includegraphics[width=8cm]{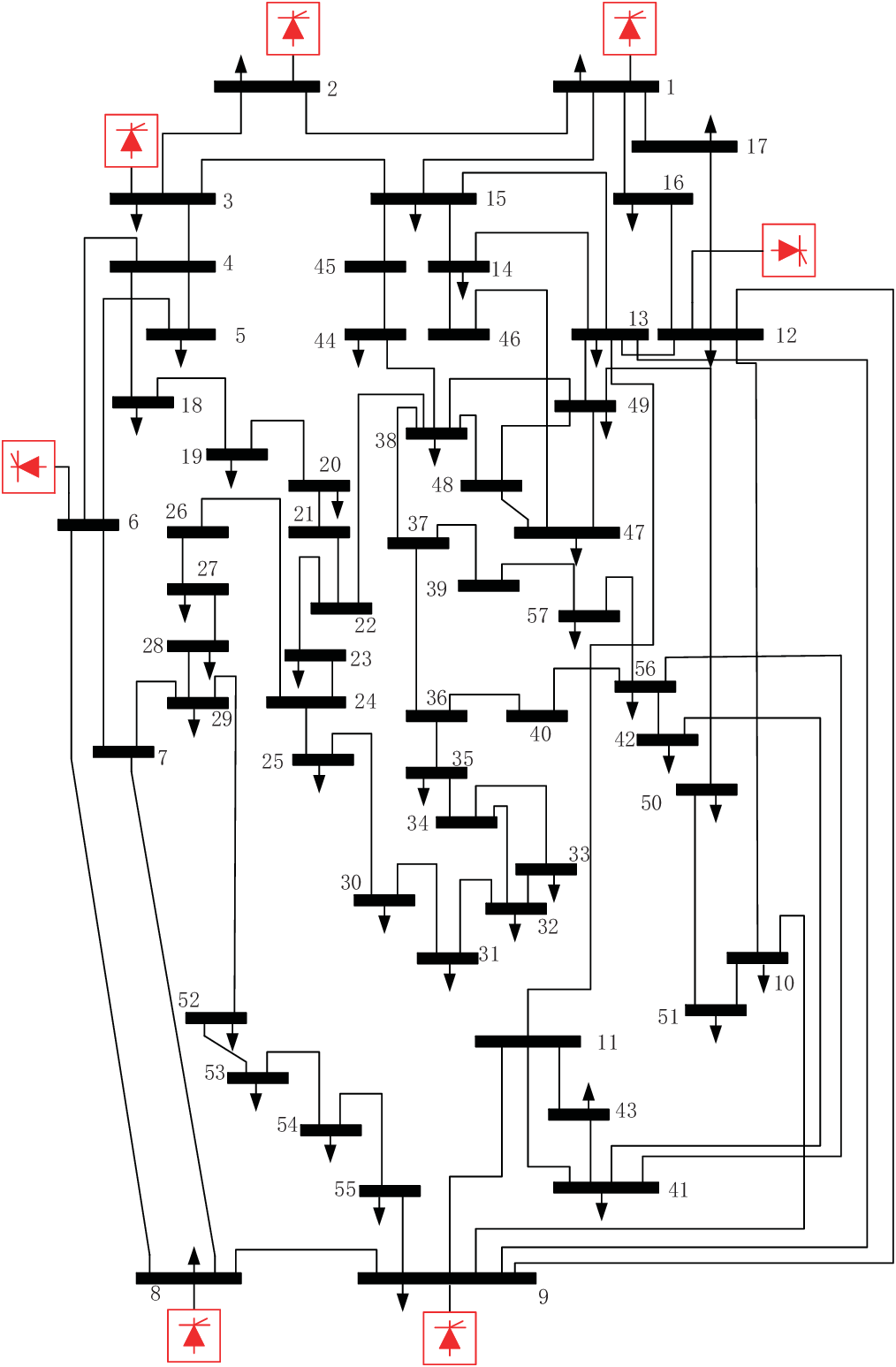} \caption{A modified IEEE-57 bus system.\label{fig6t}}
\end{figure}
\subsection{Case 6. The Test on A Large-scale Power System}
\par In order to test the performance of the designed algorithm applied to wide area systems, an modified IEEE-57 bus system as shown in Fig. \ref{fig6t} is selected to perform simulation, where the detailed configuration of this power system can be found in \cite{PSEA2023}. This system consists of 7 distributed generators, each of which is replaced by a BESS, and has 42 load nodes. The simulation results are shown in Fig. \ref{figc6}.
\par From the simulation results, it can be seen that MCs and the estimated values of average power mismatch can be reached, respectively, where MCs are able to converge to the electricity price of UG, and the estimated average power mismatch eventually converges to 0. Seven BESSs and a UG can ultimately achieve stable output, achieving power supply-demand balance. It is worth noting that compared to Case 1, the convergence speed of the algorithm on a wide area power system is slowed down. Anyway, the involved scheme is effective in addressing the distributed ED problem of a wide area power system.
\section{Conclusions}
\par To address the ED problem of a grid-connected BESS network, we design a distributed ED scheme with a PI+Reset controller. The gain conditions under which the reset mechanism acts are theoretically analyzed, and the effectiveness of this scheme is demonstrated. The effectiveness and progressiveness of the scheme are verified by a case study, and the performance of the scheme is tested in terms of different gains, time phased pricing environment, load switching, plug and play, and a large-scale system. The simulation results show that this scheme is effective, and compared to the scheme with a P controller, our designed scheme can improve convergence speed \textcolor{blue}{by reducing residence time and consensus time by 70.6\%} and control accuracy with almost no overshoot. When dealing with issues such as time phased electricity prices, plug and play, and a large-scale system, this scheme also demonstrates such advantages. However, there is still a virgin land worth striving for. In order to achieve the goal of supply-demand balance, a centralized reset mechanism is introduced into the average power mismatch estimator. It would be very meaningful if this restriction could be further relaxed.
\begin{figure}
	\centering
	\includegraphics[width=8cm]{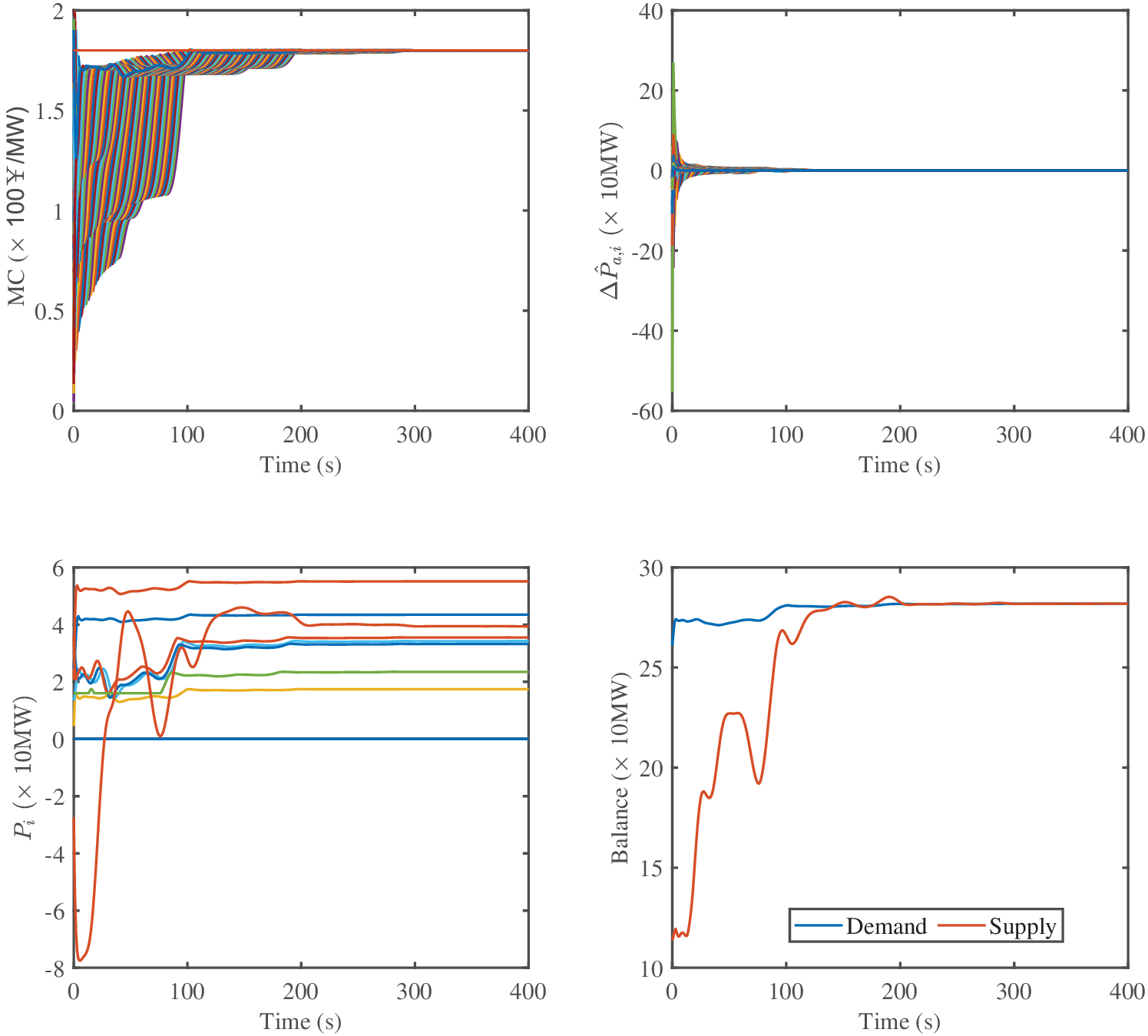} \caption{Simulation results on a large-scale power system in Case 6.\label{figc6}}
\end{figure}
\appendices
\ifCLASSOPTIONcaptionsoff
  \newpage
\fi

\bibliographystyle{IEEEtran}
\bibliography{re1.bib}
\vspace{-15 mm} 
\end{document}